\documentclass{IEEEtran}

\usepackage[utf8]{inputenc}
\usepackage[T1]{fontenc}

\usepackage{cite, hyperref, url}
\usepackage{amsmath,amsfonts,amssymb,amsthm}
\newtheorem{theorem}{Theorem}
\newtheorem{corollary}[theorem]{Corollary}
\newtheorem{lemma}[theorem]{Lemma}
\theoremstyle{remark}
\newtheorem{remark}{Remark}

\newcommand{\KL}[2]{D(#1 \| #2)}
\newcommand{\tvdist}[2]{D_{TV}(#1, #2)}
\newcommand{\expec}[2][]{\mathbb{E}_{#1}\left[ #2 \right]}
\DeclareMathOperator*{\essinf}{ess\,inf}
\DeclareMathOperator*{\esssup}{ess\,sup}

\begin{document}
\title{A Note on Reverse Pinsker Inequalities}

\author{Olivier Binette\\Universit\'e du Qu\'ebec \`a Montr\'eal\\{\small Email: \url{olivier.binette@gmail.com}}}
\IEEEpubid{\begin{minipage}{\textwidth}\ \\[12pt] \centering
 10.1109/TIT.2019.2896192 \copyright 2019 IEEE. Personal use is permitted, but republication/redistribution requires IEEE permission.\\
  See \url{http://www.ieee.org/publications\_standards/publications/rights/index.html} for more information.
\end{minipage}} 
\maketitle

\begin{abstract}
A simple method is shown to provide optimal variational bounds on $f$-divergences with possible constraints on relative information extremums. Known results are refined or proved to be optimal as particular cases.
\end{abstract}


\section{Introduction}

This note is concerned with optimal upper bounds on relative entropy and other $f$-divergences in terms of the total variation distance and relative information extremums. When taking relative entropy as the $f$-divergence, such upper variational bounds have been referred to as \textit{reverse Pinsker inequalities} \cite{Sason2016, Bocherer2016}. They are used in the optimal quantization of probability measures \cite{Bocherer2016} and have also appeared in Bayesian nonparametrics for controlling the prior probability of relative entropy neighbourhoods (see e.g. (A.2) in \cite{Binette2018}). 

Our main theorem demonstrates a simple method that yields optimal ``reverse Pinsker inequalities'' for any $f$-divergence. This refines or shows the optimality of previously best known inequalities while avoiding arguments that are tuned to particular cases. 
In particular, Simic \cite{Simic2009b} uses a global upper bound on the Jensen function to bound relative entropy by a function of relative information extremums. Corollary \ref{cor:m_M} below refines their inequality to best possible. More recently, three different bounds on relative entropy involving the total variation distance have been proposed in Theorem 23 of \cite{Sason2016} in Theorem 7 of \cite{Verdu2014} and in Theorem 1 of  \cite{Sason2015}. Our results show that the inequalities of \cite{Sason2016} and \cite{Verdu2014} are in fact optimal in related contexts. Another direct application of the method improves Theorem 34 in \cite{Sason2016}, which is an upper bound on R\'enyi's divergence in terms of the variational distance and relative information maximum, while providing a simpler proof for this type of inequality. Vajda's well-known ``range of values theorem'' (see \cite{Vajda1972, Liese2006, Vajda2009, Kumar2004, Kumar2005}) is also recovered as an application.

The rest of the paper is organized as follows. Section \ref{sec:results} presents the definitions and main results. Examples with particular $f$-divergences are provided in section \ref{sec:applications} and proofs are given in section \ref{sec:proofs}.

\section{Main results}\label{sec:results}

Let $(P, Q)$ be a pair of probability measures defined on a common measurable space. It is assumed throughout that $P \ll Q$. Given a convex function $f: [0, \infty) \rightarrow (-\infty, \infty]$ such that $f(1) = 0$, the $f$-divergence between $P$ and $Q$ is defined as
\begin{equation}\label{eqdef:f_divergence}
    D_f(P\| Q) = \expec[Q]{f\left(\frac{dP}{dQ}\right)}.
\end{equation}
In particular, the relative entropy $\KL{P}{Q}$ and the total variation distance $\tvdist{P}{Q} = \sup_{A}|P(A) - Q(A)|$ correspond to the cases $f(t) =  t\log(t)$ and $f(t) = \tfrac{1}{2}|t-1|$ respectively.

For fixed $\delta\geq 0$, $m \geq 0$ and $M \leq \infty$, we consider the set $\mathcal{A}(\delta, m, M)$ of all probability measure pairs $(P, Q)$ defined on a common measurable space and respecting the conditions : $P \ll Q$,
\begin{equation}\label{eqdef:m_M}
    \essinf \frac{dP}{dQ} = m,\quad \esssup \frac{dP}{dQ} = M
\end{equation}
and 
\begin{equation}\label{eqdef:delta}
    \tvdist{P}{Q} = \delta.
\end{equation}
Here $\essinf$ and $\esssup$ represent the essential infimum and supremum taken with respect to $Q$.

The following theorem provides the best upper bound on the $f$-divergence over the class $\mathcal{A}(\delta, m, M)$ determined by \eqref{eqdef:m_M} and \eqref{eqdef:delta}.

\begin{theorem}\label{thm:m_M_delta}
	If $\delta\geq 0$, $m \geq 0$ and $M < \infty$ are such that $\mathcal{A}(\delta, m, M) \not = \emptyset$, then
	\begin{equation}\label{eq:thm_m_M_delta}
		\sup_{(P, Q) \in \mathcal{A}(\delta, m, M)} D_f(P\| Q) = \delta \left( \frac{f(m)}{1-m} + \frac{f(M)}{M-1} \right).
	\end{equation}
\end{theorem}

\begin{remark}
    In the case where $m=1$ or $M=1$, any $(P, Q) \in \mathcal{A}(\delta, m, M)$ must be such that $\delta = D_{TV}(P, Q) = 0$. The right hand side of \eqref{eq:thm_m_M_delta} is then understood as being equal to $0$.
\end{remark}

We can obtain from Theorem 1 tight bounds for more general families of distributions. Consider for instance
\begin{equation}
    \mathcal{B}(m, M) = \bigcup_{\delta \geq 0} \mathcal{A}(\delta, m, M)
\end{equation}
and
\begin{equation}\label{eqdef:B_m_M}
    \mathcal{C}(\delta) = \bigcup_{\substack{m \in [0,1]\\M \in [1,\infty]}} \mathcal{A}(\delta, m, M).
\end{equation}

Using the first family, Corollary \ref{cor:m_M} below provides the range of $D_f$ as a function of relative information bounds. 
\begin{corollary}\label{cor:m_M}
    If $m \geq 0$ and $M < \infty$ are such that $\mathcal{B}(m, M) \not = \emptyset$, then
    \begin{equation}
        \sup_{(P, Q) \in \mathcal{B}(m, M)} D_f(P\| Q) =  \frac{(M-1)f(m) + (1-m)f(M)}{M-m}.
    \end{equation}
\end{corollary} 

Using the second family \eqref{eqdef:B_m_M}, we re-obtain Theorem 4 of \cite{Sason2016} (see also Lemma 11.1 in \cite{Basu2011}). Taking the union over possible values of $\delta$ also yields Vajda's well-known ``range of values theorem'' (see \cite{Liese2006, Vajda1972, Vajda2009, Kumar2004, Kumar2005}).
\begin{corollary}\label{cor:delta}
    If $ 0 \leq \delta \leq 1$, then
    \begin{equation}\label{eqcor:delta}
        \sup_{(P, Q) \in \mathcal{C}(\delta)} D_f(P \| Q) = \delta\left( f(0) + \lim_{M \rightarrow \infty} \frac{f(M)}{M} \right).
    \end{equation}
\end{corollary}

\begin{remark}
    Theorem \ref{thm:m_M_delta} generalizes Theorem 23 in \cite{Sason2016} with $f(t) = t\log(t)$ for the relative entropy: the upper bounds obtained are the same in this case. The proofs also share similarities. A decomposition equivalent to \eqref{eq:t1p_decomp} is used in \cite{Sason2016} and their proof is concluded by using the monotonicity of the function $t \mapsto t\log(t)/(1-t)$, continuously extended at $0$ and $1$.
\end{remark}

\section{Examples}\label{sec:applications}

This section lists applications to particular $f$-divergences and follows the standard definitions of \cite{Sason2016}. The bounds obtained are compared to similar inequalities recently shown in the literature.

\subsection{Relative entropy (Kullback-Leibler divergence)} The relative entropy corresponds to $f(t) = t\log(t)$ in \eqref{eqdef:f_divergence} and is denoted $\KL{P}{Q}$. The results are more neatly stated in this case as functions of $a = \essinf \frac{dQ}{dP} = M^{-1}$ and $b = \esssup \frac{dQ}{dP} = m^{-1}$, assuming both quantities are well defined. Theorem \ref{thm:m_M_delta} then shows
\begin{equation}\label{eq:KL_bound} 
    \sup_{(P, Q) \in \mathcal{A}(\delta, m, M)}\KL{P}{Q} 
    = \delta \left(\frac{\log(a)}{a-1} + \frac{\log(b)}{1-b}\right).
\end{equation}
In particular, the resulting upper bound on $\KL{P}{Q}$ is Theorem 23 of \cite{Sason2016}. Letting $b \rightarrow \infty$ gives the related Theorem 7 of \cite{Verdu2014} and the inequality presented therein is consequently optimal over $\bigcup_{0 \leq m \leq 1}\mathcal{A}(\delta, m, M)$.

Also, Corollary \ref{cor:m_M} yields
\begin{equation*}
    \sup_{(P, Q) \in \mathcal{B}(m, M)}\KL{P}{Q} 
    = \frac{(a-1)\log(b)+(1-b)\log(a)}{b-a}.
\end{equation*}
For comparison, Theorem I of \cite{Simic2009b} (which also appears as Theorem I in \cite{Simic2011} and is related to results in \cite{Simic2009, Simic2009a}) provides the weaker upper bound
\begin{equation*}
    \frac{a\log(b) -b\log(a)}{b-a}
            + \log\left( \frac{b-a}{\log(b)-\log(a)} \right) - 1
\end{equation*}
on $\KL{P}{Q}$ over $(P, Q) \in \mathcal{B}(m, M)$ as an application of their ``best possible global bound'' for the Jensen functional.

\subsection{Hellinger divergence of order $\alpha$} Let $\alpha \in (0,1) \cup (1, \infty)$ and $f(t) = (t^{\alpha}-1)/(\alpha-1)$. The corresponding divergence is denoted $\mathcal{H}_{\alpha} 
(P \| Q)$. Theorem 1 shows in this case
\begin{equation*}
    \sup_{(P, Q) \in \mathcal{A}(\delta, m, M)} \mathcal{H}_\alpha (P \| Q) = \frac{\delta}{1-\alpha}\left( \frac{1-m^\alpha}{1-m} - \frac{M^\alpha - 1}{M-1} \right).
\end{equation*}
When $\alpha = 2$, $\mathcal{H}_{\alpha} = D_{\chi^2}$ is the $\chi^2$ divergence and the above can be rewritten as
\begin{equation*}
    \sup_{(P, Q) \in \mathcal{A}(\delta, m, M)} D_{\chi^2} (P \| Q) = \delta(M-m).
\end{equation*}
For comparison, Example 6 of Theorem 5 in \cite{Sason2016} is the weaker inequality
\begin{equation*}
    D_{\chi^2}(P \| Q) \leq 2 \delta \max\{M-1, 1-m\}.
\end{equation*}

\subsection{R\'enyi's divergence}
Also related is R\'enyi's $\alpha$-divergence, defined as 
\begin{equation*}
    D_\alpha(P \| Q) = \frac{1}{\alpha-1}\log(1 + (\alpha -1)\mathcal{H}_\alpha(P \| Q))
\end{equation*}
and which is a monotonous transform of $\mathcal{H}_\alpha$. Correspondingly we obtain
\begin{equation*}
    D_\alpha(P \| Q) \leq \frac{1}{\alpha - 1} \log\left(1 + \delta\left( \frac{M^\alpha - 1}{M-1} - \frac{1-m^\alpha}{1-m} \right) \right).
\end{equation*}
Taking $m=0$ recovers Theorem 34 of \cite{Sason2016}. Their inequality, which is also appears in Theorem 3 of \cite{Sason2015c} for $\alpha > 2$, is improved when $m > 0$.

\section{Proofs}\label{sec:proofs}

The starting point of our analysis is the following simple known application of convexity.

\begin{lemma}\label{lem:convexity}
    Let $\kappa$ be a random variable with values in a bounded interval $I = [a, b]$, let $\varphi : I \rightarrow (-\infty, \infty]$ be a convex function and let $\bar \alpha = (b - \expec{\kappa})/(b-a)$. Then
    \begin{equation}
        \expec{\varphi(\kappa)} \leq \bar \alpha \varphi(a) + (1-\bar \alpha) \varphi(b).
    \end{equation}
\end{lemma}
\begin{proof}
    Let $\alpha$ be the non-negative random variable such that $\kappa = \alpha a + (1-\alpha) b$. Then $\mathbb{E}[\alpha] = \bar \alpha$ and by convexity of $\varphi$ we find
    \begin{align*}
        \expec{\varphi(\kappa)} 
        &\leq \expec{\alpha \varphi(a) + (1-\alpha) \varphi(b)}\\
        &= \bar \alpha \varphi(a) + (1-\bar \alpha) \varphi(b).
    \end{align*}
\end{proof}
As a particular case, we obtain a bound on the total variation distance that is of use in the proof of Theorem \ref{thm:m_M_delta}.

\begin{corollary}\label{cor:tv_bound}
    If $m \geq 0$, $M < \infty$ and $(P,Q) \in \mathcal{B}(m, M)$, then
    \begin{equation}
        \tvdist{P}{Q} \leq \frac{(M-1)(1-m)}{M-m}.
    \end{equation}
\end{corollary}
\begin{proof}
    Lemma \ref{lem:convexity}, applied with $\kappa = \frac{dP}{dQ}$, $\varphi(x) = |x-1|$, $a = m$ and $b = M$, shows that
    \begin{align*}
        \frac{1}{2}\expec[Q]{\left|\frac{dP}{dQ} - 1\right|} 
        &\leq \frac{1}{2}\left\{ \frac{M-1}{M-m}|m-1| + \frac{1-m}{M-m}|M-1| \right\}\\
        &= \frac{(M-1)(1-m)}{M-m}.
    \end{align*}
\end{proof}

We now proceed with the proof of Theorem \ref{thm:m_M_delta}.

\begin{proof}[Proof of Theorem \ref{thm:m_M_delta}]
    Let $(P, Q) \in \mathcal{A}(\delta, m, M)$. If $A = \left\{x \mid \frac{dP}{dQ}(x) \leq 1 \right \}$, then $\delta = Q(A) - P(A)$ and we may write
    \begin{multline}\label{eq:t1p_decomp}
        D_f(P \| Q) = Q(A) \expec[Q]{f\left( \frac{dP}{dQ} \right)\middle | A} \\+ Q(A^c) \expec[Q]{f\left( \frac{dP}{dQ} \right)\middle | A^c}.
    \end{multline}
    To bound the first term on the right-hand side of \eqref{eq:t1p_decomp}, note that $\expec[Q]{\frac{dP}{dQ}\middle | A} = \frac{P(A)}{Q(A)}$ and that $x \in A$ implies $m \leq \frac{dP}{dQ}(x) \leq 1$. An application of Lemma \ref{lem:convexity}, using the fact that $f(1)=0$, therefore yields
    \begin{align}
        \expec[Q]{f\left( \frac{dP}{dQ} \right)\middle | A}
        &\leq \frac{1- \frac{P(A)}{Q(A)}}{1-m} f(m)\nonumber \\
        &= \frac{\delta f(m)}{Q(A)(1-m)}. \label{eq:t1p_first_bound}
    \end{align}
    The second term is similarly bounded as to obtain
    \begin{align}\label{eq:t1p_second_bound}
         \expec[Q]{f\left( \frac{dP}{dQ} \right)\middle | A^c}
         & \leq \left(1-\frac{M-\frac{P(A^c)}{Q(A^c)}}{M-1}\right) f(M)\nonumber \\
         &= \frac{\delta f(M)}{Q(A^c)(M-1)}.
    \end{align}
    Together with \eqref{eq:t1p_decomp}, the inequalities \eqref{eq:t1p_first_bound} and \eqref{eq:t1p_second_bound} show that
    \begin{equation}\label{eq:t1p_upper_bound}
        D_f(P \| Q) \leq \delta \left( \frac{f(m)}{1-m} + \frac{f(M)}{M-1} \right)
    \end{equation}
    whenever $(P, Q) \in \mathcal{A}(\delta, m, M)$.
    
    We now show that the supremum of \eqref{eq:thm_m_M_delta} equals this bound.
    If $\delta = 0$, then the upper bound given by \eqref{eq:thm_m_M_delta} is zero and the supremum trivially attains this bound. If $\delta = 1$, then $\mathcal{A}(\delta, m, M) = \emptyset$ and the statement of Theorem \ref{thm:m_M_delta} is trivially satisfied.
    We can therefore assume $0 < \delta < 1$. Let $q = \frac{M-1}{M-m}$, $p = m q$, $t = \delta (M-m)[(M-1)(1-m)]^{-1}$ and consider the pair of discrete measures
    \begin{equation}\label{eq:prob_measures_ternary}
        \left\{\begin{array}{l}
         P = (tp, t(1-p), 1-t), \\
         Q = (tq, t(1-q), 1-t). 
        \end{array}\right.
    \end{equation}
    Corollary \ref{cor:tv_bound} ensures $0 < t < 1$ and thus $P$ and $Q$ are probability measures. It is also straightforward to verify that $(P, Q) \in \mathcal{A}(\delta, m, M)$ with $t(q-p) = \delta$, $p/q = m$ and $(1-p)/(1-q) = M$. Some algebraic manipulations then show
    \begin{align*}
        D_f(P\| Q) &= tqf\left(\frac{p}{q}\right) + t(1-q)f\left(\frac{1-p}{1-q}\right)\\
        & = \delta \left( \frac{f(m)}{1-m} + \frac{f(M)}{M-1} \right).
    \end{align*}
\end{proof}

\begin{proof}[Proof of Corollary \ref{cor:m_M}]
    Combining Corollary \ref{cor:tv_bound} with equation \eqref{eq:thm_m_M_delta} of Theorem \ref{thm:m_M_delta} yields the upper bound
    $$
        D_f(P\|Q) \leq \frac{(M-1)f(m)+(1-m)f(M)}{M-m}
    $$
    for every $(P, Q) \in \mathcal{B}(m, M)$. 
    To see that the supremum over $\mathcal{B}(m, M)$ equals this bound, it suffices to let $\delta \rightarrow (M-1)(1-m)/(M-m)$ in \eqref{eq:thm_m_M_delta}.
\end{proof}

\begin{proof}[Proof of Corollary \ref{cor:delta}]
    Some care has to be taken when considering the elements of $\mathcal{A}(\delta, 0, \infty)$. To see that the right-hand side of \eqref{eqcor:delta} also upper bounds the elements of this set, we again use the decomposition \eqref{eq:t1p_decomp}. The first term is treated as in \eqref{eq:t1p_first_bound}. For the second term, let $\frac{dP}{dQ} \wedge K = \min\{\frac{dP}{dQ}, K\}$. By Fatou's lemma and Lemma \ref{lem:convexity}, using that $f(1)=0$,
    \begin{align*}
        \expec[Q]{f\left( \frac{dP}{dQ} \right)\middle | A^c}
        &\leq \liminf_{K \rightarrow\infty} \expec[Q]{f\left( \frac{dP}{dQ}\wedge K \right)\middle | A^c}\\
        &\leq  \liminf_{K \rightarrow\infty} \frac{ \expec[Q]{\frac{dP}{dQ}\wedge K\middle | A^c} - 1}{K-1} f(K).
    \end{align*}
    By the monotone convergence theorem, 
    $$\lim_{K \rightarrow\infty}\expec[Q]{\frac{dP}{dQ}\wedge K\middle | A^c} =  \frac{P(A^c)}{Q(A^c)}
    $$ 
    and hence
    \begin{equation*}
        \expec[Q]{f\left( \frac{dP}{dQ} \right)\middle | A^c} \leq \frac{\delta}{Q(A^c)} \lim_{M \rightarrow \infty}\frac{f(M)}{M-1}.
    \end{equation*}
    We note that $\lim_{M\rightarrow \infty} \frac{f(M)}{M-1}$ exists by convexity of $f$ and can be infinite.
    The required upper bound on $D_f(P \| Q)$ is then obtained as in the proof of Theorem \ref{thm:m_M_delta}.
    
    To see that the supremum equals this bound, it suffices to let $M \rightarrow \infty$ in Theorem \ref{thm:m_M_delta}.
\end{proof}

\section*{Acknowledgment}
The author would like to thank Alexis Langlois-R\'emillard and Jean-Fran\c{c}ois Coeurjolly for helpful comments and suggestions. 
%The support of the Natural Sciences and Engineering Research Council of Canada (NSERC), through an Alexander-Graham-Bell Canada graduate scholarship, is gratefully acknowledged.

\nocite{Csiszar1967, Vajda2009, Taneja2004, Ali1966, Sason2015}

\bibliographystyle{IEEEtran}
\bibliography{reverse_pinsker_inequalities}

%\begin{IEEEbiographynophoto}
%{Olivier Binette} is a M.Sc. student in Statistics at Université du Québec à Montréal. He was granted undergraduate research awards by the \textit{Natural Sciences and Engineering Research Council of Canada} and is currently supported by an Alexander-Graham-Bell Canada graduate scholarship. His main research interests are related to Bayesian and nonparametric statistics.
%\end{IEEEbiographynophoto}

\end{document}